\documentclass[letterpaper, 10 pt, conference]{ieeeconf}  

\IEEEoverridecommandlockouts                              

\let\labelindent\relax
\usepackage[shortlabels]{enumitem}
\usepackage{graphics} 
\usepackage{epsfig} 
\usepackage{times} 
\usepackage{amsmath,amssymb,amsfonts}  
\usepackage{mathtools}
\usepackage{balance}

\newtheorem{lemma}{Lemma}
\newtheorem{proposition}{Proposition}
\newtheorem{assumption}{Assumption}
\newtheorem{definition}{Definition}
\newtheorem{remark}{Remark}

\newcommand{\interior}[1]{%
  {\kern0pt#1}^{\mathrm{o}}%
}

\title{\LARGE \bf
On A Notion of Stochastic Zeroing Barrier Function
}

\author{Tua A. Tamba$^{1}$, Bin Hu$^{2}$, and Yul Y. Nazaruddin$^{3}$
\thanks{$^{1}$Dept. Electrical  Engineering, Parahyangan Catholic University, Bandung 40141, Indonesia
        {\tt\small ttamba@unpar.ac.id}}%
\thanks{$^{2}$Dept. Engineering Technology, Old Dominion University, Norvolk, VA 23529, USA
        {\tt\small bhu@odu.edu}}%
        \thanks{$^{3}$Instrumentation \& Control Research Group, Institut Teknologi Bandung, Bandung 40132, Indonesia
        {\tt\small yul@tf.itb.ac.id}}%
}

\begin{document}

\maketitle
\thispagestyle{empty}
\pagestyle{empty}

\begin{abstract}
This note examines the safety verification of the solution of Ito's stochastic differential equations (SDE) using the notion of stochastic zeroing barrier function (SZBF).
It is shown that an extension of the recently developed zeroing barrier function concept in deterministic systems can be derived to provide an SZBF based  safety verification method for It\'o's SDE sample paths. 
The main tools in the proposed method include It\'o's calculus and stochastic invariance concept.
\end{abstract}


\section{Introduction}
The fast developments and advances in sensor, computational and communication technologies have recently stimulated a growing interests in cyber-physical systems (CPS) framework for control systems design and implementation purposes  \cite{lee2008cyber,kim2012cyber, tipsuwan2003control}.
In such a framework, a network of computers are often used to automatically manage the plant-sensor-controller-actuator interactions and data exchange via dedicated or suitable communication networks.
The resulting CPS thus often contains tight interactions between physical, computational and communication processes, resulting in complex dynamics that are often not well understood.
Among others, one of the frequently encountered challenges in CPS design are concerning the safety implementation of such CPS to guarantee the fulfillment of their safety-critical operational requirements.
While a large number of results have so far been established, issues related to safety verification of CPS remain open problems.

In order to examine and verify the safety property of CPS, researchers in systems and control areas have recently proposed a framework based on the use of the so-called \emph{barrier function} (BF). 
Being essentially a \emph{certificate-based} strategy akin with the well-known Lyapunov's stability analysis method, the BF approach verifies the safety property of CPS by finding some invariant set over which the considered safety property is fulfilled \cite{kong2013exponential, kong2014new, dai2017barrier}.
When compared to the more conventional simulation- or reachability-based methods, the BF based methods are arguably more computationally efficient as they do not require the exhaustive enumerations or simulations of the systems trajectories/paths for deciding the safety properties in questions. 
These have become the motivating reasons for the currently active research for further developments on both theoretical and application aspects of the BF methods in such areas as deterministic \cite{prajna2007framework}, hybrid \cite{prajna2004safety}, and stochastic systems \cite{prajna2004stochastic} (see also e.g. \cite{ames2019control} and the references therein).

With regard to  stochastic systems, the development of BF based safety verification method was first formulated  in \cite{prajna2004stochastic} as an exit problem.
More specifically, the  safety property is defined as the probability of the sample paths of It\'o SDE leaving a predefined safe set when initialized from a subset of that safe set.
The approach in  \cite{prajna2004stochastic} essentially search for a BF in the form of a \emph{supermartingale} of the process' sample path that can be used to upper bound such an exit probability.
The supermartingale requirement in \cite{prajna2004stochastic} for the BF  was later relaxed in \cite{steinhardt2012finite} whereby it is only required to be a $c-$martingale (cf. e.g. \cite{kushner1972stochastic}).
Akin to the idea for finite-time stability characterization developed in \cite{kushner1966finite}, such a relaxation is shown in \cite{steinhardt2012finite} to be useful for characterizing the finite-time regional safety properties of It\'o type SDE sample paths.

The aim of this note is to examine an extension of the recently developed  BF based method introduced in \cite{ames2016control} to the stochastic systems' case.
Specifically, we propose a stochastic analogue of the \emph{zeroing barrier function} (ZBF) based method developed in \cite{ames2016control} to allows for the safety verification of the solution of It\'o's SDE to be done using what we refer to in this paper as a stochastic ZBF (SZBF).
We show in this note that, similar to the development of stochastic stability methods based on Lyapunov's stability analysis in deterministic systems, the extension of ZBF based method in deterministic systems can be extended to derive a SZBF based method for safety verification of It\'o's SDE.

Section \ref{sec:setup}  formulates the problem setup and  the considered SDE model.
Section \ref{sec:result} presents the main result of the paper regarding SZBF-based invariance and stability property analyses of It\'o's SDE.
Section \ref{sec:conclude} concludes the note.

\vspace{5 pt}
\noindent {\bf Notation:}
${\mathbb R}$ and ${\mathbb R}^n$ denote the set of real numbers and the $n$-dimensional Euclidean space, respectively.
We use $x\in {\mathbb R}^n$ to denote a real-valued $n$-dimensional vector $x$ with an Euclidean norm $|x|$.
The expected value of a random variable and the probability of a random event to occur are denoted as ${\mathbb E}[\cdot]$ and ${\mathbb P}[\cdot]$, respectively.
The family of all continuous strictly increasing functions $\kappa:[0,a)\rightarrow [0,\infty)$ for some $a>0$ is denoted as class $\mathcal{K}$ function.
A continuous function $\alpha:(-b,a)\rightarrow(-\infty,\infty)$ is said to belong to the \emph{extended class} $\mathcal{K}$ function $\mathcal{K}_e$ if it is strictly increasing and satisfies $\alpha(0)=0$.
The family of all continuous functions  $\gamma:[0,b)\times[0,\infty)\rightarrow[0,\infty)$  is denoted as class $\mathcal{KL}$ function for some $b>0$ if for each fixed $s$, the mapping $\gamma(r,s)$ belongs to the class $\mathcal{K}$ function with respect to $r$ and for each fixed $r$, the mapping $\gamma(r,s)$ is decreasing with respect to $s$ and $\gamma(r,s)\rightarrow 0$ as $s\rightarrow\infty$.
$C^2_c$ denotes the family of all functions $V(x):{\mathbb R}^n\rightarrow{\mathbb R}$ that are continuously twice differentiable in $x$ with compact support.


\section{Setup \& Preliminaries}
\label{sec:setup}

\subsection{System Description}
\label{sec:desc}
Assume a complete probability space $\left(\Omega,\mathcal{F},\mathbb{P}\right)$ and consider a standard $\mathbb{R}^m-$valued Wiener process $w_t$ defined on this space.
Let $\left\{\mathcal{F}_t\right\}_{t\in\mathbb{R}_+}$ be the right continuous filtration generated by $w_t$, i.e. $\mathcal{F}_t:=\sigma(w_s;\; 0\leq s\leq t)\vee \mathcal{N}$ in which $\mathcal{N}$ denotes the class of all $\mathbb{P}-$negligible sets.
We consider a stochastic process $x_t$ which evolves on this space  according to the stochastic differential equation (SDE) of the form
\begin{equation}
\label{eq:sde}
dx_t = b(x_t)\,dt + \sum_{k=1}^m\sigma(x_t)\,dw_t^k
\end{equation}
with initial value $x_0$ at time $t=0$.
In \eqref{eq:sde}, both $b(\cdot)$ and $\sigma_k(\cdot)$ with $b(0)=0,\,\sigma_k(0)=0$ are functional mappings from $\mathbb{R}^n$ into  $\mathbb{R}^n$ which satisfy Assumption \ref{assume1} below.\\

\begin{assumption}
\label{assume1}
For $b(\cdot)$ and $\sigma_k(x)$ in \eqref{eq:sde}:
\begin{enumerate}
\item there exists a nonnegative constant $L$ such that for all $x\in\mathbb{R}^n$, the following holds.
\begin{equation}
\label{eq:asm1}
|b(x)|^2 + \sum_{k-1}^m|\sigma_k(x)|^2 \leq L\left(1+|x|^2\right)
\end{equation}
\item for any $x,y\in\mathbb{R}^n$, the following holds.
\begin{equation}
\label{eq:asm2}
|b(x)-b(y)| + \sum_{k=1}^m|\sigma_k(x)=\sigma_k(y)|\leq |x-y|\\
\end{equation}
\end{enumerate}
\end{assumption}
\vspace{5pt}

\noindent Under Assumption \ref{assume1}, a unique strong solution of the SDE \eqref{eq:sde} is known to exist in  It\^{o}'s sense  and is given by \cite{oksendal2013stochastic, khasminskii2011stochastic}
\begin{equation}
\label{eq:ito-sol}
x_t = x_0 + \int_0^t b(x_s)ds + \sum_{k=1}^m \int_0^t \sigma_k(x_s)dw_s^k
\end{equation}
where $x_0\in \mathbb{R}^n$ is given.
In what follows, for any $s\geq 0$ and $x\in\mathbb{R}^n$, we use $x^{s,x}_t$ to denote the solution of \eqref{eq:sde} of the form \eqref{eq:ito-sol} at time $s\leq t$ when initialized from $x$.

To the solution $x_t$ in \eqref{eq:ito-sol}, we associate the infinitesimal generator which is defined by an operator $\mathcal{L}$ acting on a function $h(x_t):\mathbb{R}^m\rightarrow\mathbb{R}^m$ of the form
\begin{equation}
\label{eq:gen0}
\mathcal{L}h(x_t) :=\lim_{t\searrow 0}\frac{\mathbb{E}[h(x_t)] - h(x)}{t}
\end{equation}
Here, we consider  $h(x)$ to be twice continuously differentiable on $x$ with compact support (denote this class of functions as $C^2_c$) such that \eqref{eq:gen0} becomes \cite{khasminskii2011stochastic}
\begin{equation}
\label{eq:gen1}
\mathcal{L}h(x) =\sum_{i=1}^nb^i(x)\frac{\partial h(x)}{\partial x^i} + \frac{1}{2}\sum_{i,j=1}^n\sum_{k=1}^m\sigma_k^i(x)\sigma_k^j(x)\frac{\partial^2h(x)}{\partial x^i\partial x^j}.
\end{equation}
Furthermore, for any function $h(x)\in C^2_c$ and the solution \eqref{eq:ito-sol} of the SDE \eqref{eq:sde}, It\'{o}'s lemma \cite{ito1944109} states that the following holds for $h(x_t)$.
\begin{equation}
\label{eq:ito-lemma}
dh(x_t) = \frac{\partial h}{\partial t}\,dt + \frac{\partial h}{\partial x}\,dx_t + \frac{1}{2} dx_t' \frac{\partial h(x)}{\partial x}\frac{\partial h'(x)}{\partial x}dx_t,
\end{equation}
where $dx_t$ is defined as in \eqref{eq:sde}, while $dt\,dt=0,\, dt\,dw_t^k = 0$ and $dw_t^{k^1}dw_t^{k^2}=\delta_{k^1k^2}\,dt$ for any $k^1,k^2\in k$ with $\delta$ being a dirac delta function, are used as a convention  for \eqref{eq:ito-lemma}.

\subsection{Problem Formulation}
\label{sec:prob}
Given the SDE in \eqref{eq:sde}, the objective of this paper is to examine the safety verification of the strong solution of \eqref{eq:sde}  defined in \eqref{eq:ito-sol}.
In particular, such a verification is done by characterizing the stochastic invariance property of a particular set with respect to (wrt.) the SDE \eqref{eq:sde}.
In this regard, we consider by construction the following closed set $\mathcal{C}$ (which is a subset of $\mathbb{R}^n$) defined by a function $h(x)$.
\begin{subequations}
\label{eq:invset}
\begin{align}
\mathcal{C}&=\left\{x\in\mathbb{R}^n:\, h(x)\geq 0\right\} \label{eq:invset1}\\
\interior{\mathcal{C}}&=\left\{x\in\mathbb{R}^n:\, h(x) > 0\right\} \label{eq:invset2} \\
\partial\mathcal{C}&=\left\{x\in\mathbb{R}^n:\, h(x) = 0\right\} \label{eq:invset3}
\end{align}
\end{subequations}
in which $\partial\mathcal{C}$ and $\interior{\mathcal{C}}$ denote the boundary and interior of $\mathcal{C}$, respectively.

Given the SDE in \eqref{eq:sde} and the closed set $\mathcal{C}$ in \eqref{eq:invset}, the main objective of this paper is to establish conditions on the defining function $h(x)$ that will guarantee the set $\mathcal{C}$ to be invariant under the evolution of the SDE \eqref{eq:sde}'s strong solution in \eqref{eq:ito-sol}.
Following the framework developed in \cite{ames2016control} for characterizing the invariant set of deterministic systems dynamics, our approach is based on the stochastic version of the notion of \emph{zeroing barrier function}.
To this end, we recall in Definition \ref{def:invset} below the corresponding notion of stochastic invariant set to be used in this paper.
\begin{definition}
\label{def:invset}
 A closed subset $\mathcal{C}\subset\mathbb{R}^n$ is said to be \emph{stochastically invariant} wrt. the SDE  \eqref{eq:sde} if for every $\mathcal{F}_0-$measurable random variable $x_0$ such that $x_0\in\mathcal{C}$ almost surely (a.s.), the strong solution $x_t$ in \eqref{eq:ito-sol} satisfies $x_t\in\mathcal{C}$ for all $t\geq 0$ a.s.
\end{definition}


\section{Main Results}
\label{sec:result}
This section proposes the notion of stochastic zeroing barrier function that can be used to establish the invariance of a closed set wrt. the strong solution of SDE \eqref{eq:sde} in \eqref{eq:ito-sol}.

\subsection{Stochastic Invariance Under SZBF Existence}
In this subsection, we show that the existence of a SZBF ensure the stochastic invariance of a certain subset of $\mathbb{R}^n$ wrt. the SDE in \eqref{eq:sde}.
To begin with, we first state the following notion of \emph{stochastic zeroing barrier function} (SZBF).

\begin{definition}[SZBF]
\label{def:szbf}
Consider the SDE \eqref{eq:sde} and the set $\mathcal{C}$ defined in \eqref{eq:invset}  by a function $h(x):\mathbb{R}^n\rightarrow\mathbb{R}$ with $h(x)\in C^2_c$.
If there exists a class $\mathcal{K}_e$ function $\alpha$ and a set $\mathcal{D}$ with $\mathcal{C}\subseteq\mathcal{D}\subset\mathbb{R}^n$ such that for all $x\in\mathcal{D}$ the  following hold: 
\begin{equation}
\label{eq:szbf-cond}
\text{(i) }\mathcal{L}h(x)\geq-\alpha(h(x)),\;\text{and}\; \text{(ii)}\sum_{k=1}^m\frac{\partial h(x)}{\partial x}\sigma_k(x)=0,
\end{equation}
then $h(x)$ is a SZBF.
\end{definition}

We now present the main result of this paper which essentially states that the existence of a SZBF $h(x)$ as per Definition \ref{def:szbf} implies the invariance of a closed set $\mathcal{C}$ defined in \eqref{eq:invset} wrt. the solution of the SDE \eqref{eq:sde} in \eqref{eq:ito-sol}.

\begin{proposition}
\label{prop1}
Consider the SDE \eqref{eq:sde} and a closed set $\mathcal{C}$ in \eqref{eq:invset} defined by some function $h(x):\mathbb{R}^n\rightarrow\mathbb{R}$ with $h(x)\in C^2_c$.
If $h(x)$ is a SZBF defined on the set $\mathcal{D}$ with $\mathcal{C}\subseteq\mathcal{D}\subset\mathbb{R}^n$, then $\mathcal{C}$ is stochastically invariant wrt. the solution of   \eqref{eq:sde}.
\end{proposition}
\begin{proof}
By It\'o's lemma, then the function $h(x)$ satisfies
\begin{align}
\label{eq:p1}
dh(x) &=
\frac{\partial h(x)}{\partial x} dx_t + \frac{1}{2}(dx_t)'\frac{\partial^2h(x)}{\partial x^2}(dx_t) \notag\\
&=\frac{\partial h(x)}{\partial x}\left[  b(x_t)\,dt + \sum_{k=1}^m\sigma(x_t)\,dw_t^k \right] \notag\\
&\quad+ \frac{1}{2}\left[  b(x_t)\,dt + \sum_{k=1}^m\sigma(x_t)\,dw_t^k \right]'\frac{\partial^2h(x)}{\partial x^2}\notag\\
&\hspace{50pt}\times \left[  b(x_t)\,dt + \sum_{k=1}^m\sigma(x_t)\,dw_t^k \right]
\end{align}
Expanding the right-hand side of \eqref{eq:p1} and using the convention as stated following \eqref{eq:ito-lemma}, we have that
\begin{align}
\label{eq:p2}
    dh &= \sum_{i=1}^nb^i(x)\frac{\partial h(x)}{\partial x^i} \,dt +  \sum_{k=1}^m\frac{\partial h(x)}{\partial x}\sigma_k(x)\,dw_t^k \notag \\
    &\quad+ \frac{1}{2}\sum_{i,j=1}^n\sum_{k=1}^m\sigma_k^i(x)\sigma_k^j(x)\frac{\partial^2h(x)}{\partial x^i\partial x^j}\,dt  \notag \\
    & = \left(\sum_{i=1}^nb^i(x)\frac{\partial h(x)}{\partial x^i} + \frac{1}{2}\sum_{i,j=1}^n\sum_{k=1}^m\sigma_k^i\sigma_k^j(x)\frac{\partial^2h(x)}{\partial x^i\partial x^j} \right)dt \notag \\
    & \quad + \sum_{k=1}^m\frac{\partial h(x)}{\partial x}\sigma_k(x)\,dw_t^k
\end{align}
which, by the construction in \eqref{eq:gen1}, simplifies to
\begin{equation}
\label{eq:dh0}
dh(x)=\mathcal{L}h(x)\,dt + \sum_{k=1}^m\frac{\partial h(x)}{\partial x}\sigma_k(x)\,dw_t^k.
\end{equation}
Now if $h(x)$ is a SZBF, then it must satisfies condition \eqref{eq:szbf-cond} in Definition \ref{def:szbf}.
This  implies that \eqref{eq:dh0} may be rewritten as
\begin{equation}
\label{eq:dh}
dh(x)\geq -\alpha(h(x))\,dt
\end{equation}
for some $\alpha\in\mathcal{K}_e$. 
Using similar argument as in the proof of \cite[Proposition 1]{ames2016control}, we have $\dot{h}(x)\geq-\alpha(h(x))=0$ for any $x\in\partial\mathcal{C}$. By Nagumo's theorem \cite{aubin1995stochastic, da2004invariance,blanchini1999set}, we conclude that the set $\mathcal{C}$ in \eqref{eq:invset} is stochastically invariant.
\end{proof}

One important aspect of the result stated in Proposition \ref{prop1} is that it allows one to determine the connection between the ZBF (when exists) of the diffusion-free part of \eqref{eq:sde}  and the SZBF of the SDE \eqref{eq:sde}.
Specifically, if the set $\mathcal{C}$ defined by  $h(x)\in C^2_c$ is invariant wrt. the diffusion-free part $dx_t = b(x_t)\,dt$ of \eqref{eq:sde}, the conditions for $\mathcal{C}$ to be stochastically invariant wrt. the SDE \eqref{eq:sde} is stated in Lemma \ref{lem1} below.

\begin{lemma}
\label{lem1}
Assume the defining function $h(x)\in C^2_c$ of the set $\mathcal{C}$ in \eqref{eq:invset}  is a ZBF for the diffusion-free  part $dx_t = b(x_t)\,dt$ of \eqref{eq:sde}.
If the diffusion part of the SDE in \eqref{eq:sde} satisfies:
\begin{enumerate}[(i)]
\item $\frac{1}{2}\sum_{i,j=1}^n\sum_{k=1}^m\sigma_k^i\sigma_k^j(x)\frac{\partial^2h(x)}{\partial x^i\partial x^j}=0$
\item $\sum_{k=1}^m\frac{\partial h(x)}{\partial x}\sigma_k(x)\,dw_t^k=0$
\end{enumerate}
then the following statements are equivalent:
\begin{enumerate}
\item $h(x)$ is a SZBF for the SDE \eqref{eq:sde}
\item $\mathcal{C}$ is invariant wrt. the diffusion-free and the SDE \eqref{eq:sde}
\end{enumerate}
\end{lemma}

\begin{proof}
Note that \eqref{eq:p2} is of the form
\begin{multline}
\label{eq:p3}
    dh(x) = \sum_{i=1}^nb^i(x)\frac{\partial h(x)}{\partial x^i}\, dt   \\
     + \frac{1}{2}\sum_{i,j=1}^n\sum_{k=1}^m\sigma_k^i\sigma_k^j(x)\frac{\partial^2h(x)}{\partial x^i\partial x^j} \,dt \\
     +\sum_{k=1}^m\frac{\partial h(x)}{\partial x}\sigma_k(x)\,dw_t^k.
\end{multline}
Since $h(x)$ is a ZBF for the diffusion-free part of \eqref{eq:sde}, we have by \cite[Proposition 1]{ames2016control} that $\dot{h}(x)\geq-\alpha(h(x))$ in which $\alpha\in\mathcal{K}_e$.
Combining this and conditions (i)-(ii) in Lemma \ref{lem1}, we have  $\mathcal{L}h(x)\geq-\alpha(h(x))$ and $\sum_{k=1}^m\frac{\partial h(x)}{\partial x}\sigma_k(x)=0$ which by Definition \ref{def:szbf} implies $h(x)$ is a SZBF for SDE \eqref{eq:sde} as claimed in statement 1) in the lemma.
The statement 2) in the lemma follows from the fact that $h(x)$ is a ZBF of the diffusion-free part of \eqref{eq:sde} (cf. \cite[Proposition 1]{ames2016control}) as well as a SZBF for the SDE  \eqref{eq:sde} (cf. Proposition \ref{prop1}).
\end{proof}

\begin{remark}
Intuitively, Lemma \ref{lem1} states that the invariance property of the differential equation of the form $\dot{x}(t)=b(x)$ may be preserved in its stochastic counterpart of the form \eqref{eq:sde} only if the diffusion part of the SDE \eqref{eq:sde} satisfies conditions (i)-(ii) in the lemma.
\end{remark}


\subsection{SZBF-induced Stochastic Stability}
In this section, we show that the existence of a SZBF for the SDE in \eqref{eq:sde} induces a stochastic Lyapunov function which guarantees the stochastic stability of the corresponding invariance set.
This thus essentially establishes the stochastic counterpart of \cite[Proposition 2]{ames2016control} that was developed for the deterministic systems case.

\begin{proposition}
\label{prop2}
Let $h(x):\mathcal{D}\rightarrow\mathbb{R}$ with $h\in C^2_c$ be defined on an open subset $\mathcal{D}\subseteq\mathbb{R}^n$ such that $\mathcal{C}\subset\mathcal{D}\subseteq\mathbb{R}^n$ where $\mathcal{C}$ is defined as in \eqref{eq:invset}.
If $h(x)$ is a SZBF wrt. the SDE in \eqref{eq:sde}, then the set $\mathcal{C}$ is stochastically stable.
\end{proposition}

\begin{proof}
Note that $h(x)$ being a SZBF  on  $\mathcal{D}$ induces a stochastic Lyapunov function $V_{\mathcal{C}}(x):\mathcal{D}\rightarrow\mathbb{R}$ of the form
\begin{align}            
V_{\mathcal{C}}(x)&=
\begin{cases}
0, 	& \text{if $x\in\mathcal{C}$}, \\[1.2ex]
-h(x), 	& \text{if $x\in\mathcal{D}\setminus\mathcal{C}$}.
\end{cases}
\end{align}
With such a choice, one may notice that  $V_{\mathcal{C}}(x)$ is continuous on its domain and $V_{\mathcal{C}}(x)\in C^2_c$ at every point $x\in\mathcal{D}\setminus\mathcal{C}$. 
Furthermore: i) $V_{\mathcal{C}}(x)=0$ for $x\in\mathcal{C}$; ii) $V_{\mathcal{C}}(x)>0$ for $x\in\mathcal{D}\setminus\mathcal{C}$; and iii) for $x\in\mathcal{D}\setminus\mathcal{C}$, then $\mathcal{L}V_{\mathcal{C}}(x)$ satisfies
\begin{align}
\label{eq:LV}
\mathcal{L}V_{\mathcal{C}}(x) &= - \mathcal{L}h(x) \notag\\
&\leq\alpha\circ h(x) = \alpha(-V_{\mathcal{C}}(x))\leq 0,
\end{align}
with $\alpha\in\mathcal{K}_e.$ 

Let us consider a solution $x_t^{0,x_0}:=x_0\in\mathcal{I}_0\subseteq\mathcal{C}$ of \eqref{eq:sde} in which $\mathcal{I}_t$ denotes the zero set of the function $V_{\mathcal{C}}(x_t)$ at time $t$ (i.e. $x_0$ belongs to the zero set of $V_{\mathcal{C}}(x)$ at $t=0$). Assume that $\mathcal{I}_t$ is closed in $\mathcal{C}$ such that $\mathcal{I_0}$ is also closed in $\mathcal{C}$, then by the continuous everywhere property of \eqref{eq:sde} in \eqref{eq:ito-sol}, there exists a time $t_{\mathcal{C}}$ such that $x_{t_{\mathcal{C}}}^{0,x_0}\in\mathcal{C}$ for a certain time interval $[0,\,t_{\mathcal{C}})$ and that $x_{t_{\mathcal{C}}}^{0,x_0}$ is also a strong solution of \eqref{eq:sde}.
In this regard, $t_{\mathcal{C}}>0$ with probability 1 is the first instance when  $x_{t_{\mathcal{C}}}^{0,x_0}$ leaves the set $\mathcal{C}.$
Using \eqref{eq:p2} to compute It\'o formula for $V_{\mathcal{C}}(x)$ on the interval $[0,\,t_{\mathcal{C}})$, we have that
\begin{align*}
V_{\mathcal{C}}(x(t_{\mathcal{C}}\wedge t)) - V_{\mathcal{C}}(x_0) &= \int_0^{t_{\mathcal{C}}\wedge t}\mathcal{L}V_{\mathcal{C}}(x_s)\, ds\notag \\
&{}- \sum_{k=1}^m\int_0^{t_{\mathcal{C}}\wedge t} \frac{\partial V_{\mathcal{C}}(x)}{\partial x}\sigma_k(x)\,dw_s^k\,ds
\end{align*}
Taking the expectation of the above equation gives
\begin{align}
\mathbb{E}\left[V_{\mathcal{C}}(x(t_{\mathcal{C}}\wedge t))\right] - V_{\mathcal{C}}(x_0) &=\mathbb{E}\left[ \int_0^{t_{\mathcal{C}}\wedge t}\mathcal{L}V_{\mathcal{C}}(x_s)\, ds\right] 
\end{align}
which by \eqref{eq:LV} implies that
\begin{equation}
\label{eq:LV2}
\mathbb{E}\left[V_{\mathcal{C}}(x(t_{\mathcal{C}}\wedge t))\right] \leq 0. 
\end{equation}
Combining \eqref{eq:LV2}  with the property that $V_{\mathcal{C}}=0$ for $x\in\mathcal{C}$ and positive  for $x\in\mathcal{D}\setminus\mathcal{C}$, we then have that
\begin{equation}
\label{eq:LV2}
V_{\mathcal{C}}(x(t_{\mathcal{C}}\wedge t))= 0 
\end{equation}
with probability 1 which thus implies that $(t_{\mathcal{C}}\wedge t):=t$ holds with probability 1.
Then by \cite[Lemma 7.4]{khasminskii2011stochastic}, we conclude that the set $\mathcal{C}$ is stable in probability.
\end{proof}


\section{Remark and Discussion}
\label{sec:conclude}
This paper has presented an approach for the safety verification of the solution of Ito's stochastic differential equations using the notion of stochastic zero barrier function.
It is shown that the extension of ZBF based method in deterministic systems can be extended to provide SZBF based method for safety verification of It\'o's stochastic differential equation. \addtolength{\textheight}{-12cm}   



\balance

%
%
\section*{Acknowledgment}
This research was  supported by Ministry of Research and Technology - National Research and Innovation Agency (Kemenristek-BRIN) of the Republic of Indonesia under the Fundamental Research Scheme, 2020.

%



\end{document}